\newif\ifusenix
\newif\ifacm
\newif\ifmcom
\newif\ifieee

\usenixfalse
\acmtrue
\mcomfalse
\ieeefalse
\ifieee
\documentclass[10pt, conference, letterpaper]{IEEEtran}
\IEEEoverridecommandlockouts
\fi
\ifusenix
	\documentclass[letterpaper,twocolumn,10pt]{article}
	\usepackage{usenix}
	\usepackage{times}
\fi
\ifacm
  \documentclass[10pt,sigconf]{acmart}
  
\fi

\ifmcom
  \documentclass{sig-alternate-10pt}
\fi

\usepackage{booktabs} %added by Ish for nice looking tables
\usepackage{pifont}% http://ctan.org/pkg/pifont for tick and cross symbols
\usepackage{xcolor}
\usepackage{amsfonts}
\usepackage{balance}
\PassOptionsToPackage{hyphens}{url}
\usepackage{hyperref}
\usepackage{color}
\usepackage{graphics}
\usepackage{graphicx}
\usepackage{url}
\usepackage{listings}
\usepackage{multicol}
\usepackage{multirow}
\usepackage[scaled]{helvet}
\usepackage{rotating}
\usepackage{xspace}
\usepackage{algorithm}
\usepackage{comment}
\usepackage{enumitem}
\usepackage{amsmath}
\usepackage{mathrsfs}
\usepackage{cancel}
\usepackage{cleveref}
\usepackage{subfigure}
\usepackage{commath}

\usepackage{amssymb} % used for \blacksquare
\usepackage[font=small,labelfont=bf]{caption}
\usepackage{amsthm}
\theoremstyle{plain}

\newcommand{\newtodo}[1]{\ClassWarning{NOT READY TO SUBMIT}{There is something left todo} \textcolor{blue}{}}

\newcommand{\name}{FlexLink\xspace}
\newcommand{\dpa}{DPA\xspace}



\definecolor{deepblue}{rgb}{0,0,0.5}
\definecolor{deepred}{rgb}{0.6,0,0}
\definecolor{deepgreen}{rgb}{0,0.5,0}
\definecolor{backcolour}{rgb}{0.95,0.95,0.92}

\def\beq{\begin{equation}}
\def\eeq{\end{equation}}
\def\beqa{\begin{eqnarray}}
\def\eeqa{\end{eqnarray}}
\def\beqan{\begin{eqnarray*}}
\def\eeqan{\end{eqnarray*}}

\newtheorem{theorem}{$\blacksquare$ Theorem}

\renewcommand{\qed}{\rule{0.7em}{0.7em}}

\def\qed{\mbox{} \hfill $\Box$}
\def\tm1{t\! - \! 1}
\def\tp1{t\! + \! 1}

\def\ant{\text{ant}}

\begin{document}

\acmYear{2025}\copyrightyear{2025}
\setcopyright{cc}
\setcctype[4.0]{by}
\acmConference[MobiHoc '25]{International Symposium on Theory, Algorithmic Foundations, and Protocol Design for Mobile Networks and Mobile Computing}{October 27--30, 2025}{Houston, TX, USA}
\acmBooktitle{International Symposium on Theory, Algorithmic Foundations, and Protocol Design for Mobile Networks and Mobile Computing (MobiHoc '25), October 27--30, 2025, Houston, TX, USA}
\acmDOI{10.1145/3704413.3764470}
\acmISBN{979-8-4007-1353-8/25/10}
%%%%%%%%%%% Agrim commands-2 start%%%%%%%%%%%%
% \interfootnotelinepenalty=10000
% \setlength{\belowdisplayskip}{2pt} \setlength{\belowdisplayshortskip}{2pt}
% \setlength{\abovedisplayskip}{2pt} \setlength{\abovedisplayshortskip}{2pt}
%%%%%%%%%%%%% end %%%%%%%%%%%%%%%%%%%%
% \title{Supporting multi-user with multi-beam multi-frequency delay-based phased array design}
\title{ FlexLink: Decoupling Control and Data Beams for Next-Generation Wideband Networks
% Decoupling Control \& Data Beams in Wideband Multi-antenna Networks
% Concurrent Communication and Initial Access for mmWave network with a single RF Chain
% Towards Flexible Antenna Array for mmWave (Joint) Communication and Sensing Beyond 5G
}

% \author[]{Paper \#607 12 pages + References + Appendix}

% Ish Kumar Jain
\author{Ish Kumar Jain}
\orcid{0000-0002-6321-379} % optional, replace with real ORCID
\affiliation{%
  \institution{Rensselaer Polytechnic Institute}
  % \department{Department of Electrical, Computer, and Systems Engineering}
  \city{Troy}
  \state{NY}
  \country{USA}}
% \affiliation{%
%   \institution{University of California San Diego}
%   \department{Department of Electrical and Computer Engineering}
%   \city{La Jolla}
%   \state{CA}
\email{jaini@rpi.edu}

% Rohith Reddy Vennam
\author{Rohith Reddy Vennam}
\affiliation{%
  \institution{University of California San Diego}
  % \department{Department of Electrical and Computer Engineering}
  \city{La Jolla}
  \state{CA}
  \country{USA}}
\email{rvennam@ucsd.edu}

% Dinesh Bharadia
\author{Dinesh Bharadia}
\affiliation{%
  \institution{University of California San Diego}
  % \department{Department of Electrical and Computer Engineering}
  \city{La Jolla}
  \state{CA}
  \country{USA}}
\email{dineshb@ucsd.edu}

% \renewcommand{\shortauthors}{First Author et al.}

% \author{Ish Kumar Jain, Rohith Reddy Vennam, Dinesh Bharadia}
% \affiliation{
%     \institution{University of California San Diego}
%     \city{La Jolla}
%     \state{CA}
%     \country{USA}
%     }
\renewcommand{\shortauthors}{IK Jain, RR Vennam, D Bharadia}
% \email{{ikjain, rsubbaraman, dineshb}@eng.ucsd.edu}

% \renewcommand{\shorttitle}{\name: Decoupling Control and Data Communication}


\ifacm
    % !TEX root = main.tex

\begin{abstract}

The next generation of 6G networks aims to utilize ultra-wideband spectrum and massive antenna arrays to serve multiple users with both control and data channels at low latency and high efficiency. However, phased arrays at mmWave and mid-bands are fundamentally constrained to a single beam or suffer sharp beamforming loss when split across directions, limiting simultaneous control-data support. In FlexLink, we introduce and prototype a novel delay-phased array architecture that overcomes this limitation by redistributing energy jointly across frequency and space, enabling multiple narrow beams without sacrificing per-beam gain or requiring additional power. 
We design and prototype FlexLink on a custom 4-7 GHz hardware testbed, demonstrating for the first time that control and data beams can be decoupled in practice, achieving nearly double spectral efficiency compared to conventional phased arrays.

\end{abstract}

\fi

\maketitle

\ifusenix
    
\fi

\ifmcom
    
\fi
\ifieee
    
\fi

% !TEX root = main.tex
\section{Introduction}\label{sec:intro}

% Motivation
% mmwave great. High bandwidth OFDM - heavy appl, multi antenna - high coverage. 

% % what is the problem
% Coupled control and data. 

% % why hard

% % key insights

% % Solution

% % Impact

Wideband and Multi-antenna systems such as current Millimeter wave (mmWave) bands in Frequency Range 2 (24.25 GHz to 52.6 GHz) and upcoming mid-bands in Frequency Range 3 (7.125 GHz to 24.25 GHz) are vital for next-generation 6G and beyond networks to support high throughput and low latency applications such as autonomous vehicles, industrial IoT, XR streaming. To facilitate communication at these bands, 3GPP adopts OFDMA to optimize time and frequency resources (in the form of resource blocks or RB) with various control and data signals.
% These bands have been specified in 3GPP release xyz and have already seen deployments in various cities such as xyz and are expected to grow by xyz base stations in the next xyz years. To facilitate communication at these bands, 3GPP specifies certain control and data signals with certain structure that are vital for coordination and functioning. For instance,  the initial access procedure in FR2 bands requires multiple SSB control signals for initial access. These control SSB signals adverties the base station ID to allow any new user device to connect to the base station. Ideally, we want the control and data symbols to occupy orthogonal time-frequency resources with no overlap and fill all available time-frequency resource blocks, i.e., utilize some RBs for control and other RBs for data. For instance, with a bandwidth of 400 MHz and sub-carrier spacing of 120 kHz, we could have 270 RBs in one time slot. We could achieve high spectrum utilization by using all 270 RBs for control and data signals.
\begin{figure}[t!]
    \centering
    \includegraphics[width=0.4\textwidth]{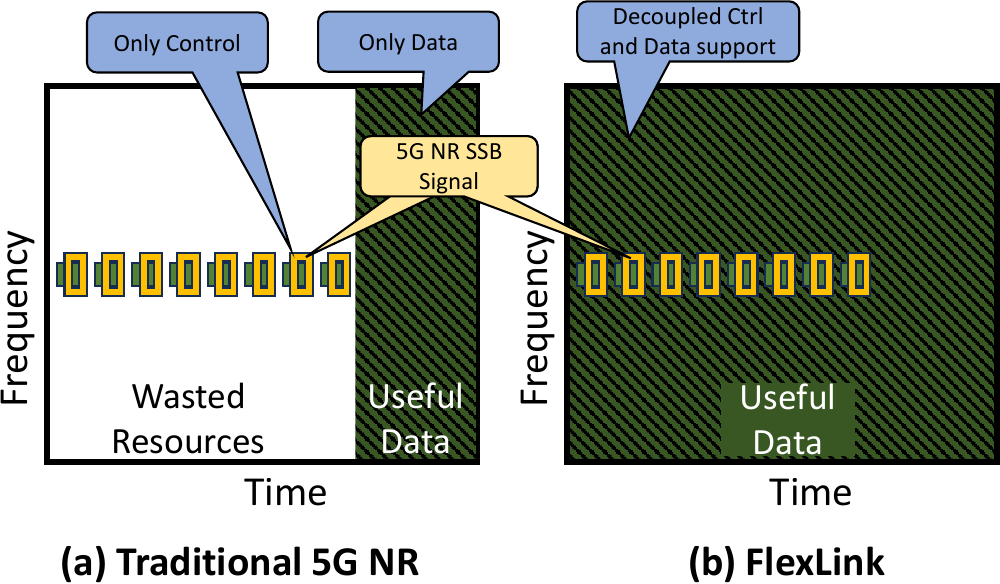}
    \caption{\name designs a novel radio architecture that can decouple control and data signaling through flexible frequency resource distribution of control and data to different directions with narrow pencil beams, much like a flexible configurable prism.
    }
    \vspace{-0.02\textheight}
    \label{fig:intro}
\end{figure}
However, current mmWave systems suffer from low spectrum utilization and cannot fill all orthogonal RBs with control and data signals due to a constraint posed by directional links. A phased array creates a single beam to radiate all RBs in a single direction, meaning all the control and data are radiated in one single direction. For instance, the SSB (Secondary Synchronization Block) control signal beam has to cover all 360 space since a new user may appear at any angle around the base station. However, the data beam focuses on an active user direction, which may not be located along the SSB control beam direction (Figure~\ref{fig:intro}). 
% Therefore active users may not be served along with the control beam.
% located uniformly in 360 space, therefore, they may not be . 
% For example, when the control beam is at 30 degrees, and there are no active users at 30 degrees, then the base station can only send control and no data signal in that time slot. 
Additionally, since these control signals occupy a small number of resource blocks, the remaining RBs could go unused/wasted resources. SSB control signal requires only 7\% of the entire 400 MHz band; the remaining 93\% RBs are unused and wasted as they cannot serve any user with data communication~\cite{5gnr}. This frequency-direction constraint leads to the coupling of control and data signals. The specs for mid-bands are not yet designed, but they are expected to follow a similar pattern to 5G NR~\cite{bazzi2025upper}, facing similar issues of coupled data and control beams.

% There are numerous studies that try to decouple the control and data signals by swaying away from traditional single beam architectures. 
To decouple control and data signals, multi-beams architectures are proposed, that can support multiple directions simultaneously, i.e., one direction for control signals and the remaining directions for data communication in a single time slot.
% , but they all suffer from performance-complexity tradeoff. High performance in terms of throughput or spectral efficiency comes at the cost of high complexity in the form of cost and power. For instance, a fully-digital architecture can support multiple parallel beams for control and data, but at the cost of higher number of high-sampling ADC/DACs and RF chains that increases with the bandwidth and number of antennas. Other efforts are to create new analog architectures with a single RF chains. 
Traditional phased arrays can be used to create multi-beams, but since they utilize an antenna-splitting mechanism~\cite{jain2021two, aykin2019smartlink}, they compromise performance in terms of lower throughput and coverage. Splitting a beam into two reduces beamforming gain by half for each split beam to preserve the total radiated power. TTD-based architecture is proposed to create infinitely many beams at the same time by spreading each frequency subcarrier to different directions like a prism~\cite{li2022rainbow,yan2019wideband,boljanovic2021fast, jain2023mmflexible, ratnam2022joint}. But this extreme of infinitely many beams radiates in all directions, even those that may not have an active user, thus wasting spectrum resources. Moreover, each direction receives signals from a tiny fraction of bandwidth, which is not enough to schedule a wideband signal transmission for data signals. Therefore, there is a gap in the literature to provide support for decoupling control and data signals without wasting precious spectrum resources.

In this paper, we propose \name, a novel multi-beamforming system to decouple the control and data beams in a flexible, 5G-compliant manner that achieves high spectrum utilization.
\name uses an antenna array architecture called delay-phased array (DPA)~\cite{jain2023mmflexible,ratnam2022joint}. The \dpa architecture consists of both delays and phases, unlike traditional phased arrays, which only have phase elements, or TTD arrays, which only have delay elements. Recent work on \dpa shows the ability to create flexible frequency-dependent multi-beamforming, with the ability to stream a subset of frequencies in one direction and another subset in another direction until all frequency resources are utilized without wastage. While this architecture has been previously used for multi-user communication with frequency multiplexing, we propose a new application for decoupling control and data signals, while meeting more stringent requirements for this application. We particularly make three important contributions:

\textbf{1) Fast configuration of multi-beams:} First, we develop a new optimization framework to estimate delays and phases per antenna to create a desired multi-beam response, for instance, one beam for control and another for data at the same time, using orthogonal frequency RBs. Unlike previous iterative methods, which suffer from high computation complexity ($O(MN\log(MN))$ for M directions and N frequency subcarriers), our formulation delivers the same result in one shot. We achieve this fast estimation by approximating an NP-hard optimization problem as an L2 norm minimization problem with a non-linear constraint and solving it to a closed-form expression. The closed-form formula allows us to obtain per-antenna delay and phase values as a function of the desired number of beams, beam angles, and the fraction of bandwidth per beam. The computation is O($K$) for phases and O($K^2$) for delays, for $K$ beams, independent of the number of antennas $N$ or number of subcarriers $M$. In practice, we require only a few beams for control and data traffic (e.g., one beam for control beam scan and 1-2 beams for data), so the complexity is close to O(1), for a few fixed beams scenarios. To the best of our knowledge, this is the first work to achieve a closed-form formula for delay and phase values to program \dpa. 

% In addition to reducing complexity, we provide important insights derived from the closed form expression. The main 

\textbf{2) Supporting a wide range of bandwidth parts: } Second, the multi-beams should support both control and data signals with a wide range of bandwidth parts. The control signal may require as small as 7\% of bandwidth, while the data may need a large 93\%. So, the hardware should support such diverse bandwidth parts. We show that the traditional solution is limited to a minimum of 20\% bandwidth part support. Lower than 20\% bandwidth part, the beamforming gain degrades exponentially and cannot be supported. In contrast, we optimize the framework to support this wide range of bandwidth-parts while supporting less than 10\% bandwidth-part for the control beam, while the remaining 90\% goes to the data beam.

\textbf{3) Designing hardware prototype for \name: } Finally, we are the first to build a hardware prototype for \dpa and demonstrate the performance of frequency-dependent multi-beams in real-world settings with over-the-air experiments. Our prototype operates in the 4-7 GHz range, which can be extended to mid-band frequencies. We discuss various design choices in building this hardware prototype, including the impact of range and resolution of delay and phase values per antenna, the methodology to achieve a wideband delay element using a switched array architecture, and system integration. The code and artifacts for \name are available online\footnote{Artifact link ~\url{wcsng.ucsd.edu/dpa}}.

% We provide details on our PCB for control and power distribution and FPGA implementation for fast beamforming at \url{wcsng.ucsd.edu/flexlink}
% , which will be open-sourced to the community. 

% : It contains control signal distribution using Artix CMOD A7 FPGA to configure the delay and phase element as fast as in 2 us. It also has power distribution modules to supply appropriate voltage levels to all components in the system.

\textbf{\name Evaluation Overview:} We verify the performance of \name through both system and circuit-level simulations and over-the-air hardware measurements. We compare antenna gain and throughput performance with two baselines: One \dpa-based heuristic algorithm, which suffers from high computational complexity, and another phased array multi-beam baseline, which suffers from low SNR and throughput. We show that \name can create any arbitrary multi-beam response with a configurable bandwidth part with O(1) complexity, improving the spectral efficiency of both control and data beams simultaneously.

\section{Related Work} \label{sec:related}
% aykin2019multi
\textbf{Multi-beams with phased arrays:}  
Conventional phased arrays can generate multiple concurrent beams, but only by splitting power, which reduces effective gain and throughput~\cite{marjan2018split,jain2021two,aykin2019smartlink,hassanieh2018fast}. Quasi-omni designs~\cite{ismayilov2018adaptive} perform even worse, spreading energy broadly with low SNR. These methods also radiate the entire bandwidth in all directions. In contrast, \name leverages frequency-dependent multi-beams that concentrate gain in targeted directions, simultaneously supporting control and data without sacrificing performance.  

% gao2023beamsquint,li2024can,
\textbf{Fixed true-time delay (TTD) and related arrays:}  
Architectures based on TTD~\cite{yan2019wideband,boljanovic2021fast,wadaskar20213d}, frequency-scanning~\cite{li2022bringing,sekretarov2023frequency}, and leaky-wave antennas~\cite{saeidi2021thz,ghasempour2020single} create frequency-dependent patterns by dispersing signals prism-like across directions. However, each beam carries only a small bandwidth fraction ($<6\%$~\cite{boljanovic2021fast,wadaskar20213d}), wasting spectrum when no active users are present. Moreover, such designs often require modifications at both base stations and user devices, making them incompatible with 5G NR. By combining delay and phase, \name avoids these pitfalls and remains fully 5G NR compliant while flexibly assigning bandwidth per beam.  

\textbf{Delay-phased arrays (DPA):}  
Recent work on DPA~\cite{jain2023mmflexible} and joint phase-time arrays~\cite{ratnam2022joint,mo2025beamforming} shows the potential of frequency-selective multi-beams. Prior studies explored beam squint mitigation~\cite{zhao2024fast,nguyen2024joint}, 2D arrays~\cite{yildiz20243d,wadaskar2025fast}, near-field effects~\cite{cai2025hybrid}, and codebook design~\cite{wadaskar2023structured}. Yet, none derived closed-form delay/phase expressions or addressed bounded delay requirements for hardware. \name is the first to provide closed-form analysis, apply DPA to control–data decoupling, and validate performance with an open-source over-the-air prototype. A recent Samsung demonstration~\cite{mo2025beamforming} further confirms practicality, but \name uniquely demonstrates end-to-end feasibility for simultaneous control and data.  

\textbf{Hardware TTD implementations:}  
TTD hardware spans CMOS/SiGe RFICs with ns-scale ranges~\cite{mondal20172,aghazadeh20223,li2020800},
% ,jung2020compact,abbasi2025four,xu202410ns,kim2018cmos
FR2/D-band silicon with ps-scale delays~\cite{lee2019continuous},
% karakuzulu2020broadband
and photonic beamformers offering ultrabroad bandwidths~\cite{zhu2020silicon}.
% martinez2024ultrabroadband
These approaches trade delay range against operating frequency. Unlike prior device-level efforts, \name provides a system-level insight: the delay needed for multi-beam combining is bounded and independent of array size. We exploit this analytically and experimentally, enabling a practical 4–7~GHz prototype with only a 1~ns delay range using the Extreme-Waves unit~\cite{ew}.

% !TEX root = main.tex

\section{Design for \name} \label{sec:design}
\begin{figure}[t!]
    \centering
    \includegraphics[width=0.45\textwidth]{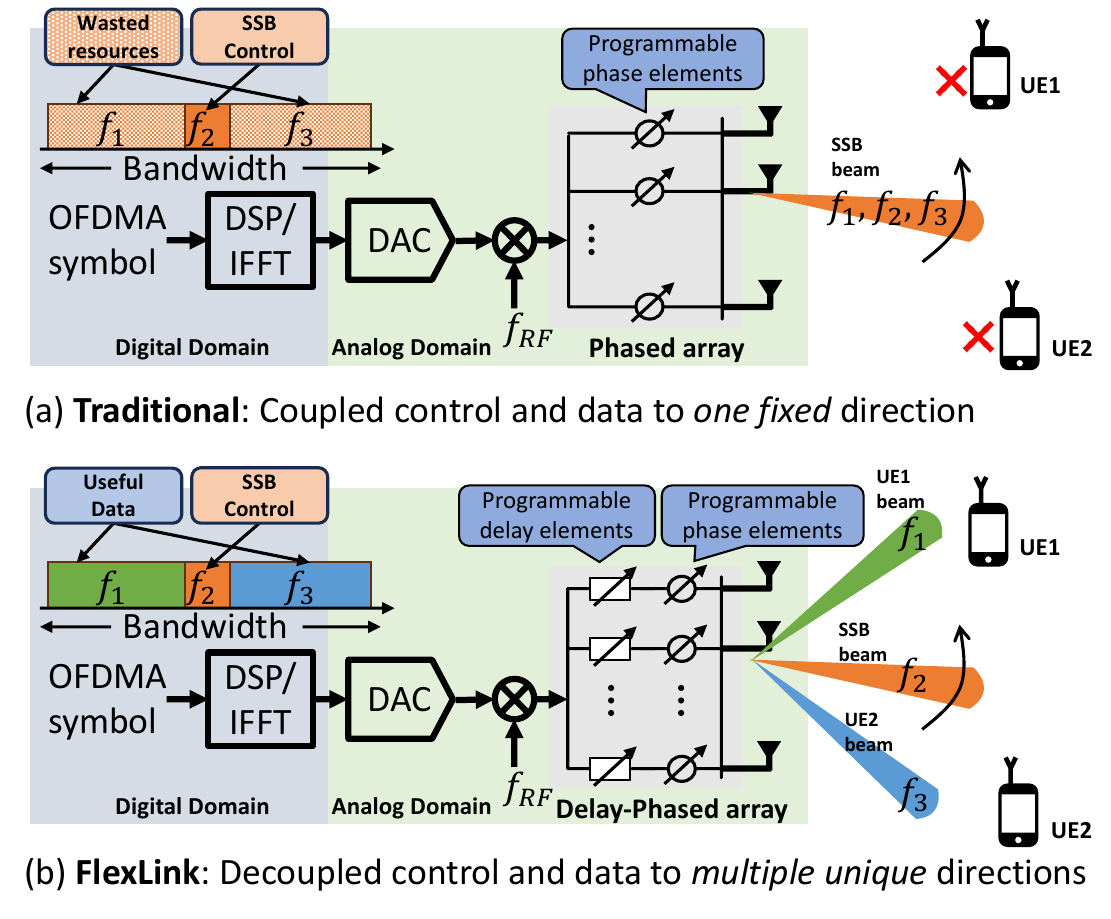}
    \vspace{-0.015\textheight}
    \caption{Delay-phased array (DPA) is an analog front-end architecture with programmable delay and phase elements. }
    \vspace{-0.02\textheight}
    \label{fig:dpa_architectures}
\end{figure}

\textit{\textbf{Problem: Coupled control and data signaling }}
In conventional phased arrays, all subcarriers are radiated in the same beam direction due to a single RF-chain analog front-end. As shown in Figure~\ref{fig:dpa_architectures}(a), three disjoint subcarrier groups $f_1$, $f_2$, and $f_3$ are steered together toward an SSB beam at $\theta_2$, even when active users are located at $\theta_1$ and $\theta_3$. This coupling wastes time–frequency resources and prevents simultaneous service of users in different directions. A naive fix is to widen or split beams, but this reduces beamforming gain~\cite{jain2021two}, degrading throughput and coverage.

\textit{\textbf{FlexLink with Delay-Phased Array:}}
\name overcomes this limitation using a delay-phased array (\dpa) architecture that integrates programmable delays and phase shifts. The delay element $\tau_n$ at antenna $n$ induces a frequency-dependent phase rotation $2\pi f \tau_n$, enabling different frequency bands to radiate in different directions. This property allows \name to decouple the SSB control beams from the data beams, for example, directing the SSB toward $\theta_2$ while simultaneously steering the data beams to $\theta_1$ and $\theta_3$ (Figure~\ref{fig:dpa_architectures}(b)). Unlike split-beam methods, this approach preserves beamforming gain and unlocks higher throughput and lower latency.

The remainder of this section is organized as follows. Section~\ref{sec:derivation} derives closed-form expressions for computing delay and phase settings to realize multi-beam patterns. Section~\ref{sec:arbitrary_bp} introduces a bandwidth-splitting strategy that supports arbitrary resource allocation per beam. Section~\ref{sec:hardware} discusses insights from our eight-antenna \dpa prototype, validating the practicality of the approach.

\subsection{Closed-form computation of delay and phase in \dpa} \label{sec:derivation}
Our goal is to program the \dpa delay and phase values to generate a desired frequency-dependent multi-beam response. The number of beams in our multi-beams can be arbitrary but finite. The direction of each beam can be anywhere in the field-of-view of \dpa, ideally at different angles. Importantly, each beam should carry a fraction of the bandwidth (as a set of contiguous subcarriers or bandwidth parts), such that the total bandwidth parts sum to the system bandwidth. These flexible multi-beams are achieved via \dpa hardware by programming the delay and phase elements appropriately. Mathematically, we model \dpa as a uniform linear antenna array with the phase value $\Phi_n$ and delay value $\tau_n$ at antenna index $n$. The antenna weight vector $w_\text{dpa}(n,f)$ is then given by:
\begin{equation}
    w_\text{dpa}(n,f) = e^{j\Phi_n+j2\pi f \tau_n}
\end{equation}
which is a function of antenna index $n$ and frequency $f$. The beamforming gain $G(f,\theta)$ of \dpa as a function of the weight vector is given by:
\begin{equation}
    G(f,\theta) = \sum_{n=0}^{N-1} w_\text{dpa}(n,f) e^{-jn\pi\sin(\theta)} \label{Gain equation}
\end{equation}
which is a function of the beamforming angle $\theta$ and frequency $f$, much like a 2D energy map shown in Figure~\ref{fig:dpa_architectures}(d). Note, we assume the antenna spacing is approximately $\lambda/2$ and ignore the effect of beam squint for simplicity~\cite{ratnam2022joint}.

Our objective is to maximize the beamforming gain along the desired multibeam directions as:
\begin{equation}
\begin{split}
    &\max_{\tau_n,\Phi_n} \;\left\|G(f,\theta)\right\|^2, \; s.t. (\theta, f) \in \{(\theta_1,f_1), \ldots, (\theta_K, f_K)\}
\end{split}
\end{equation}
where the desired $K$ multi-beam directions are represented by set $\{(\theta_1,f_1), \ldots, (\theta_K, f_K)\}$. Ideally, we do not want to radiate any energy in the undesired bands/angles, but it is inevitable due to the undesired side-lobes that appear due to a finite number of antennas.

Now, we ask what set of delays and phases per antenna would give us the beamforming gain pattern with the desired bandwidth part and beam direction. We first consider a simple case of two beams with equal bandwidth parts of $B/2$ each, where $B$ is the total system bandwidth. We assume the two beams are directed along $(-\theta_0, \theta_0)$ respectively as shown in Figure \ref{fig:dpa_proof}(a). We will later discuss a general case with an arbitrary bandwidth part and beam direction.

We formulated it as an optimization problem and solved it to a closed-form expression for the set of delays $\tau_n$ and phases $\Phi_n$ for each antenna $n$ ($n=0,1,\ldots, N-1$) that would generate the given beamforming response as follows:

\begin{theorem}\label{th:2beamcase}
\textbf{(2-beam symmetric case)}
The closed-form expression for the set of delays $\tau_n$ and phases $\Phi_n$ for each antenna $n$ ($n=0,1,\ldots, N-1$) that would generate a given two-beam response with equal bandwidth part $B/2$ and angles $\pm \theta_0$ respectively is as follows:
\begin{equation} \label{eq:tau_closed_proof}
    \tau_n = \left(\frac{3}{2B}n\sin(\theta_0) +\frac{3}{4B}\right)\;\;\;\text{mod }\frac{3}{2B} 
\end{equation}

\begin{equation}\label{eq:phase_closed_proof}
    \Phi_n = \text{round}(n\sin(\theta_0))\pi \;\;\;\text{mod }2\pi
\end{equation}

\end{theorem}

Before proving the theorem, we want to draw two important insights from this closed-form formula. First, the computation of delays and phases is like a plug-and-play model by putting the desired angle-frequency pairs into the formula. Unlike past methods, this does not require iterative optimization, thus reducing the computational complexity to O(1), allowing fast computation in FPGAs. Second, similar to how phase is bounded by $2\pi$, the delays are also bounded by a max value of $3/2B$, independent of the number of antennas. This shows the promise of scaling the \dpa architecture to arbitrarily large antenna arrays without requiring longer delay lines that are very hard to achieve in practice (discussed in Section~\ref{sec:hardware}).

\begin{figure}[t!]
    \centering
    \includegraphics[width=0.46\textwidth]{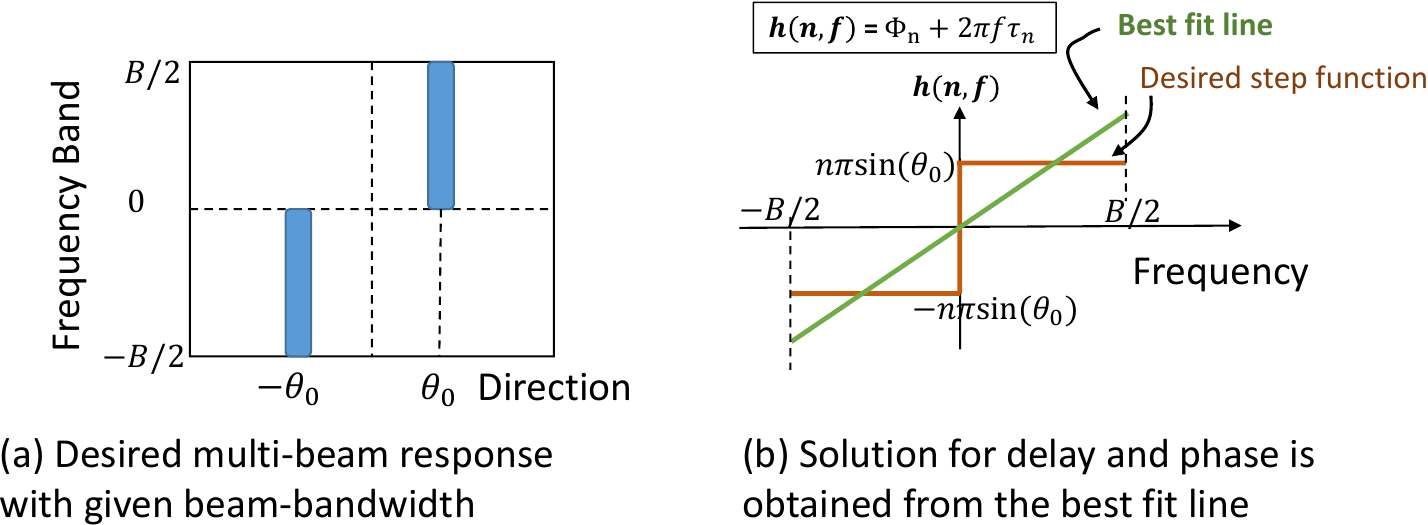}
    \caption{Proof of estimating closed-form delay and phase values for a given two-beam response at $\pm\theta_0$.}
    \vspace{-0.02\textheight}
    \label{fig:dpa_proof}
\end{figure}

\begin{figure*}[!t]
  \begin{minipage}[b]{0.18\linewidth}
    \centering
    \includegraphics[width=\linewidth]{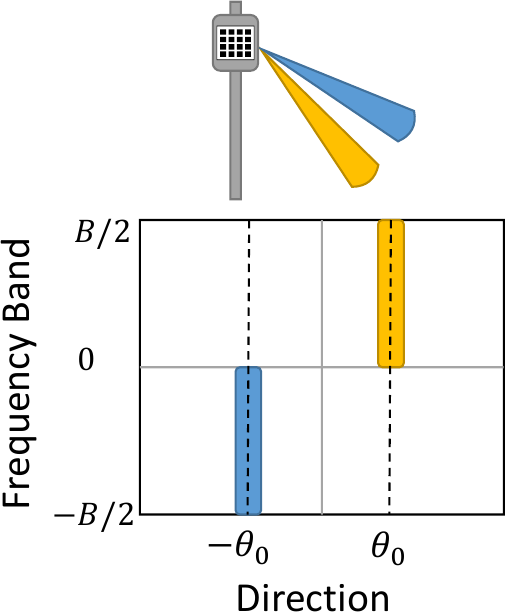}
    \caption{Desired frequency-space beam response for 2 users at directions $-\theta_0$ and $\theta_0$.}
    \label{fig:proof_desired_response}
  \end{minipage}
  \hspace{0.01\linewidth}
  \begin{minipage}[b]{0.79\linewidth}
    \centering
    \includegraphics[width=\linewidth]{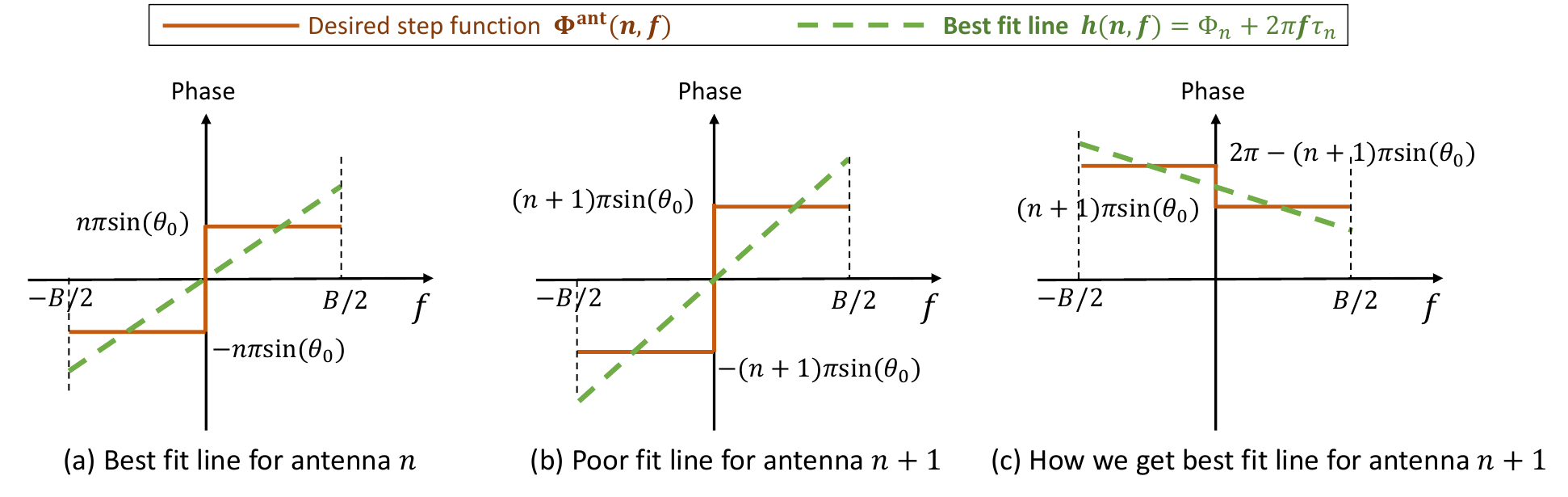}
    \caption{
    Sketch of proof: We show the phase response for each antenna is a step function with a variable step size that depends on the antenna index, $n$. 
    }
    \label{fig:proof_best_line_fit}
  \end{minipage}
\end{figure*}

\textbf{Proof of Theorem 1: Two-beam case}
We drive the expression of delays and phases per antenna that would give us the beamforming gain pattern with the desired bandwidth part and beam direction for the two-beam case.

\begin{figure} [!t]
\centering
\subfigure[Without optimization]{
    \includegraphics[width=0.22\textwidth]{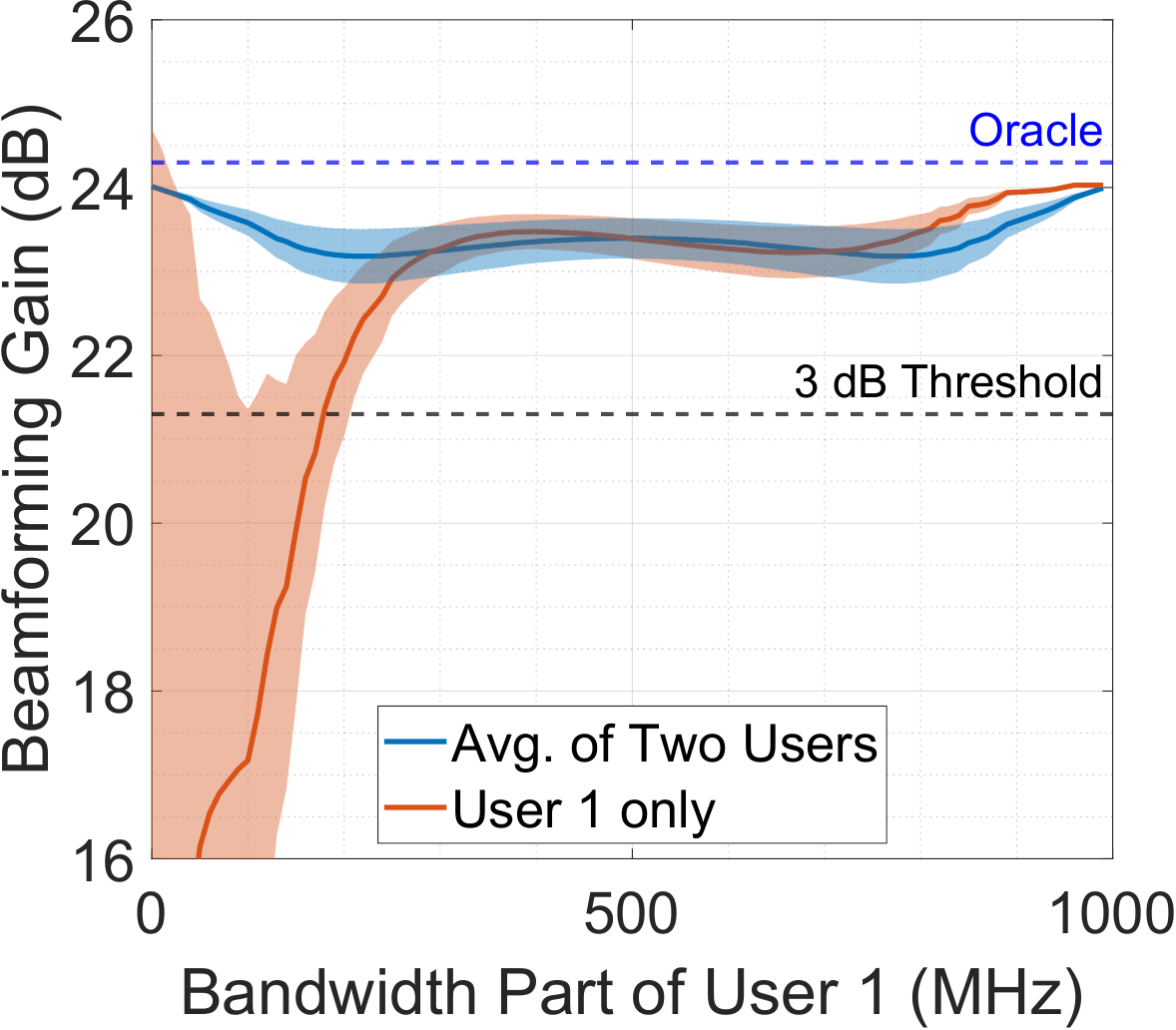}
    \label{fig:impact_bw_fraction_without_thresh}
  }\hfill
  \subfigure[With \name optimization]{
    \includegraphics[width=.22\textwidth]{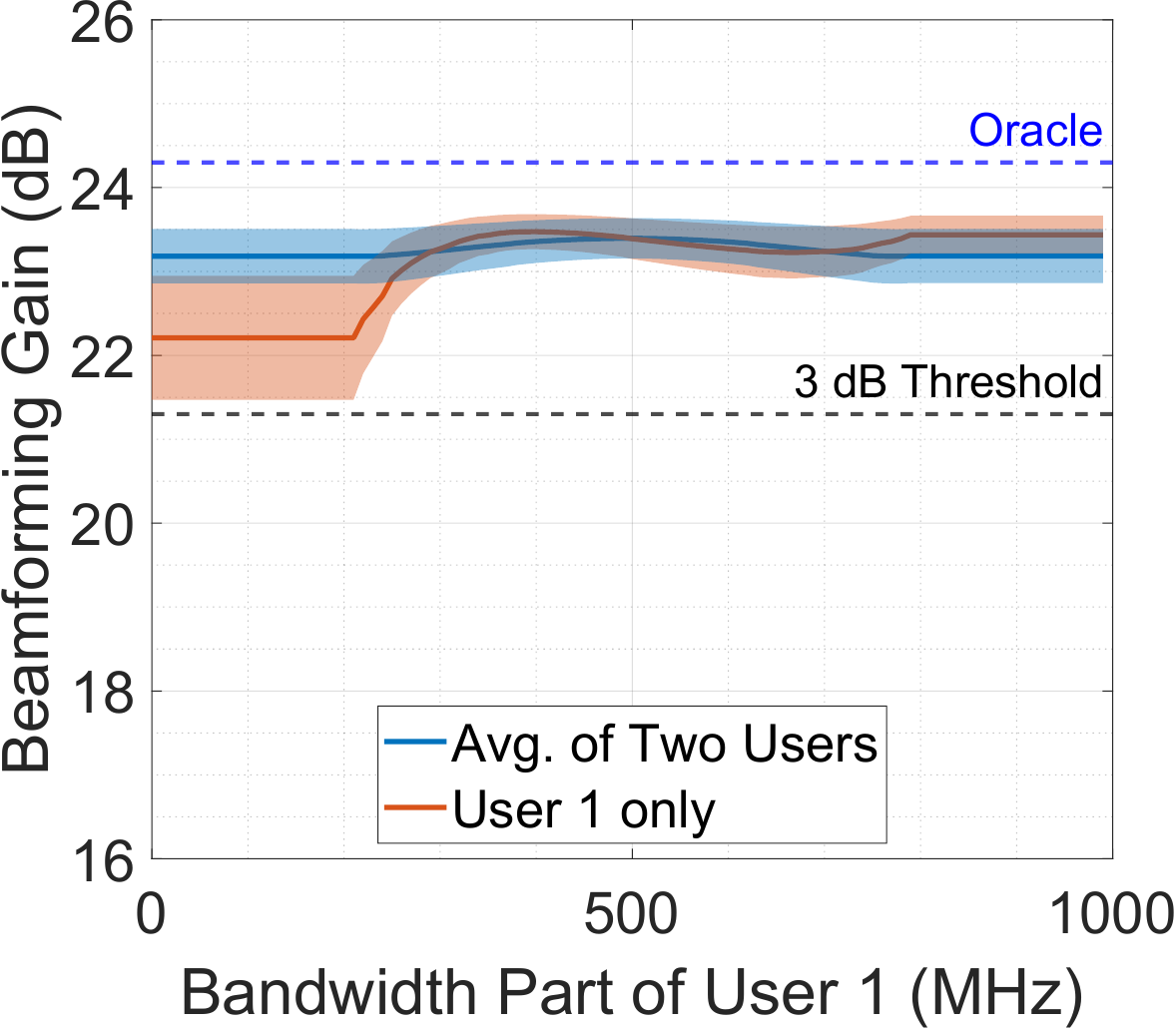}
    \label{fig:impact_bw_fraction_with_thresh}
  }      
  \vspace{-0.02\textheight}
  \caption{Impact of low bandwidth fraction and mitigation. 
    }
  \vspace{-0.02\textheight}
    \label{fig:impact_bw_fraction}
\end{figure}

We reformulate the optimization in (3) to include two-beam constraints as:

\vspace{-5pt}
\begin{equation}\label{eq:constraint}
\begin{split}
    \max_{\tau_n,\Phi_n} &\left\|\sum_{n=0}^{N-1} e^{j h(n,f)} e^{-j\Phi^\ant(n,f)}\right\|^2\\
    s.t. \quad  &h(n,f)=\Phi_n+2\pi f \tau_n\\
    \text{and}\quad &\Phi^{\text{ant}}(n,f) = \begin{cases}
       -n\pi\sin(\theta_0) & f\in[\frac{-B}{2},0]\\
       n\pi\sin(\theta_0) & f\in(0,\frac{B}{2}]\\
    \end{cases}
\end{split}
\end{equation}

where $h(n,f)$ is a function of variable phase $\Phi_n$ and delay $\tau_n$ at antenna $n$ and the function $\Phi^{\text{ant}}(n,f)$ represents the constraints from the desired frequency-direction response. 

We propose an optimization framework that can help to find a closed-form expression for delays and phases. Our optimization problem is formulated in a way that finds the line $h$ that best fits the given step function $\Phi^\ant$. We achieve this by approximating our optimization problem in (7) to the best line-fitting on a per-antenna basis:
\begin{equation} \label{eq:error_objective}
\begin{split}
        &\min_{\tau_n,\Phi_n} ||h(n,f) -  \Phi^\ant(n,f)||^2\\
\end{split}
\end{equation}
We can visualize this optimization in Figure \ref{fig:proof_best_line_fit}(a), where the line $h(n,f)$ is fit over the step function $\Phi^\ant(n,f)$. The slope of the best-fit line gives the delay value, and the y-intercept gives the phase value. In this way, we can estimate both delay and phase values by solving for the best-fit line.

\noindent
\textbf{$\blacksquare$ Impact of linear approximation:}
However, as antenna index $n$ increases, the error in line fitting also increases due to the linear increase in the step size with $n$, as shown in Figure~\ref{fig:proof_best_line_fit}(b). This could lead to high error for large antenna arrays and limit our solution to scale with antennas.

\noindent
\textbf{$\blacksquare$ Minimizing linear approximation error:}
We have an innovative and simple solution to address this issue.  To address this issue, we utilize the concept of wrapping the phase of a signal by $2\pi$, i.e., adding an integer multiple of $2\pi$ to the phase does not change the signal. We use this idea to strategically add a phase of multiple of $2\pi$ to a specific set of frequencies in order to minimize the error in line fitting as shown in Figure~\ref{fig:proof_best_line_fit}(c). With this insight, we redefine the step function $\Phi^\ant$ as:
\begin{equation}
    \Phi^{\text{ant}}(n,f) = \begin{cases}
        k2\pi-n\pi\sin(\theta_0) & f\in[\frac{-B}{2},0]\\
        n\pi\sin(\theta_0) & f\in(0,\frac{B}{2}]\\
    \end{cases}
\end{equation}
where $k$ is a constant integer. A natural question is how do we estimate this integer to minimize the error in line fitting? Our solution is a two-step process: we solve for the delays and phases as a function of $k$ and then find the optimal value of $k$ to minimize the error.

To solve for per-antenna delays and phases, we form a system of linear equations. We discretize the frequency as $f=m\Delta f$ for $m\in [-M/2, M/2]$, where the bandwidth is $B=(M+1)\Delta f$. Note that there are $M$ frequency bins that can be a large number, i.e., $M\rightarrow \infty$ for creating a continuous frequency axis. We then formulate a set of linear equations for each frequency term to solve for the variable delay $\tau_n$ and phase $\Phi_n$ for each antenna $n$, given by:
\begin{equation}
    \Phi_n + 2\pi m\Delta f\tau_n = \Phi^\ant(n,m\Delta f) \;\forall m\in [-M/2, M/2]
\end{equation}
which is in the form of a system of linear equations ($Ax=b$), with unknown $x = [\tau_n,\Phi_n]$. Solving this system of linear equations gives the desired estimate of delay and phase values.

\begin{figure*}[ht]
    \centering
    \begin{minipage}[t]{0.32\textwidth}
        \centering
        \includegraphics[width=\textwidth]{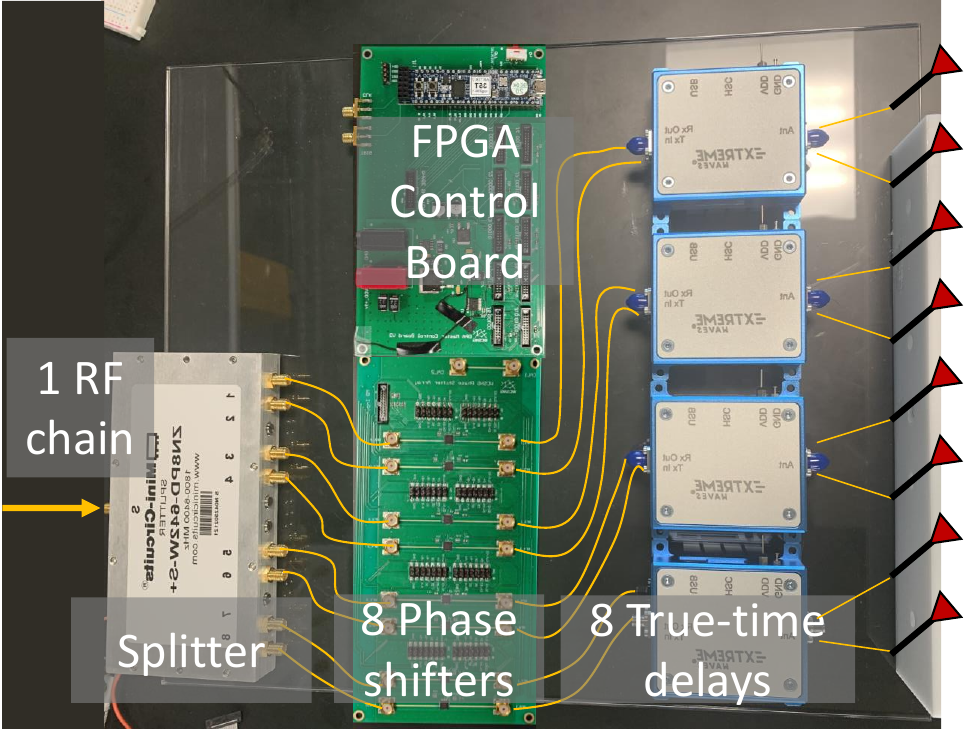}
        \caption{Implementing 8 antenna delay-phased array at 4-7 GHz band. }
        \label{fig:dpa_impl}
    \end{minipage}
    \hfill
    \begin{minipage}[t]{0.32\textwidth}
        \centering
        \includegraphics[width=\textwidth]{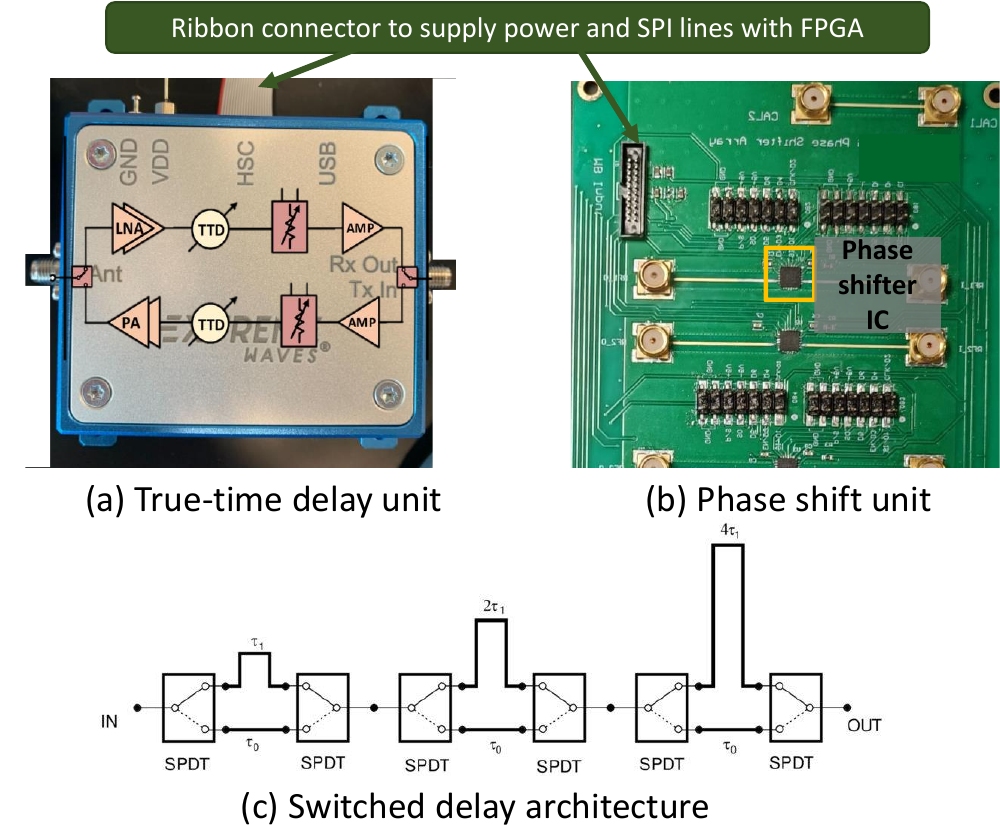}
        \caption{Hardware components of TTD unit and phase shift unit.}
        \label{fig:ttd_internal}
    \end{minipage}
    \hfill
    \begin{minipage}[t]{0.32\textwidth}
        \centering
        \includegraphics[width=\textwidth]{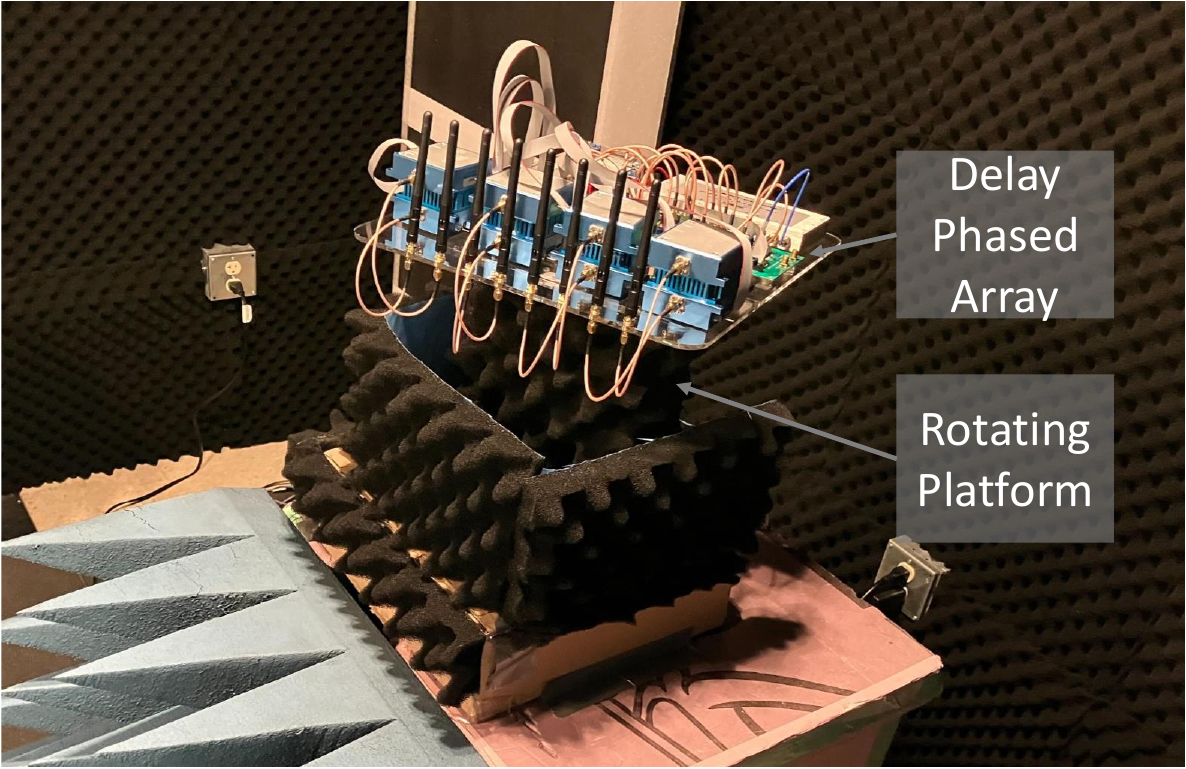}
        \caption{Beam pattern measurement in an anechoic chamber.}
        \label{fig:anechoic_chamber}
    \end{minipage}
\end{figure*}

\noindent
\textbf{$\blacksquare$ Generalization to an arbitrary number of beams:} We also show generalized beamforming response to an arbitrary number of beams with arbitrary beam directions and arbitrary bandwidth parts.

\begin{theorem}\label{th:generalized}
\textbf{(Generalized case):} Let there are $D$ beam directions with beam angles $\theta_d$ and bandwidth part $\alpha_d B$ for $\sum_d \alpha_d = 1$. We define $\phi_d = n\pi\sin(\theta_d)$ for simplicity. The per-antenna phases and delays in realizing such a generalized beamforming response is given by:
\begin{equation}\label{eq:phase_gen}
    \Phi_n = \sum_{d=1}^D\alpha_d(\phi_d+2k_d\pi)
\end{equation}
\begin{equation}\label{eq:delay_gen}
    \tau_n = \sum_{d=1}^D\frac{3}{\pi B}(\phi_d+2\pi k_d)\alpha_d(\sum_{\ell=1}^{d}2\alpha_{\ell} -\alpha_d -1)
\end{equation}
where,
\begin{equation}
    \phi_d = n\pi\sin(\theta_d)
\end{equation}
and the constant integer $k_d$ for beam $d$ and antenna $n$ is:
% \begin{equation}
%     \phi_d = n\pi\sin(\theta_d)
% \end{equation}
% and
\begin{equation}\label{eq:kd_gen}
    k_d =\begin{cases}
    &0\quad d=1\\
    &k_{d-1}\!+\!\text{round}(\frac{n\sin(\theta_{d-1})-n\sin(\theta_d)}{2}) \;\;   d\ge 2
    \end{cases}
\end{equation}
% is a constant integer.\\

\end{theorem}
The proof follows similarly to the two-beam case\footnote{Visit \url{wcsng.ucsd.edu/dpa} for details.}. Also, note that the number of multiplications in the computation of phase in (\ref{eq:phase_gen}) is O(K) and delay (\ref{eq:delay_gen}) is O($K^2$), which is independent of the number of antennas or system bandwidth, making it scalable to larger arrays.

\subsection{\textbf{Flexible bandwidth part with \dpa}}\label{sec:arbitrary_bp}
A key requirement of \name is supporting arbitrary bandwidth parts. However, line-fitting errors limit \dpa multi-beamforming, especially when a user’s bandwidth part is very small. Figure~\ref{fig:impact_bw_fraction_without_thresh} shows that while overall two-user gain remains high, users with less than 20\% of system bandwidth suffer more than 10 dB loss due to skewed step functions and poor line fitting.

To mitigate this, we introduce a threshold mechanism that prevents any beam’s gain from dropping more than 3 dB below optimal. This is enforced by setting a minimum bandwidth threshold of 20\% (corresponding to 3 dB loss), though the value can be adjusted (e.g., 25\% for 1.5 dB, 15\% for 6 dB). If a beam requests less than 20\%, we reassign delays and phases as if in a 20\%–80\% split. As shown in Figure~\ref{fig:impact_bw_fraction_with_thresh}, this caps the low-bandwidth beam’s loss at 3 dB, while the other beam experiences less than 1 dB loss, an acceptable trade-off. This heuristic thus enables support for arbitrarily small bandwidth parts without severe performance degradation.

\subsection{\textbf{\dpa Hardware Design}}\label{sec:hardware}

To demonstrate the practical performance of \name, we designed a custom prototype of \dpa hardware in the lab using commercial true-time delay chips and phase shifter chips operating at 4 GHz--6 GHz frequencies as shown in Figure~\ref{fig:dpa_impl}. Although generally, there are differences between sub-6 GHz and mmWave systems, the generation of desired beam patterns is agnostic to center frequency and can be validated with our sub-6 setup. Here, we discuss the \dpa hardware requirements in terms of range, resolution, wideband design of delay and phase circuits, and implementation setup.

\noindent
\textbf{DPA Hardware Overview:} We describe the hardware shown in Figure~\ref{fig:dpa_impl} from a downlink perspective and emphasize that the uplink path follows a similar reverse order. The input to \dpa hardware is a single analog radio frequency waveform at the desired center frequency, which can be generated from a single DAC and a series of mixers (e.g., a signal generated from a USRP). We then split this RF signal into 8 equal parts using a 1:8 splitter and pass the 8 copies of the signal to 8 phase shifters. The phase-shifted signal then passes through 8 true-time delay modules, which are finally connected to 8 antennas. Thus, the hardware resembles the concept diagram we presented in Figure~\ref{fig:dpa_architectures}.

\noindent
\textbf{$\blacksquare$ Range of delay and phase elements in \dpa:}
The range of phase values is a constant $2\pi$ that covers all scenarios because the exponential phase always wraps around $2\pi$. However, the range of delay values is not straightforward because, unlike phase, delays don't generally wrap around a certain value. In fact, prior works on using a delay element for multi-beamforming showed an unbounded delay that increases linearly with the number of antennas~\cite{li2022rainbow}, thus making it hard to build and scale in a practical circuit design because of large size, bandwidth, and matching constraints~\cite{ghaderi2019integrated}. In contrast, we designed \name to have shorter and bounded delay lines that do not scale with the number of antennas. For the two-beam case, the delay range for \name is $\frac{3}{2B}$, which is $1.5 ns$ for 1 GHz bandwidth. the range of delay can be further reduced with a smaller field of view of the array. A $120^o$ angular field-of-view requires only $1 ns$ delay range.

\noindent
\textbf{$\blacksquare$ Wideband Delay and Phase Units:} To meet the ultra-high bandwidth of 5G and the next generation of cellular networks, we need to design a wideband circuit for delay and phase units. Delay lines can be implemented in an integrated circuit using techniques such as passive elements~\cite{li2020800}, all-pass filter~\cite{mondal20172}, and Switched Capacitor design~\cite{lin2023design, ghaderi2019integrated} depending on applications and requirements. Our true-time delay unit is built in-house by Extreme Waves~\cite{ew}, which implemented delays using a switched delay architecture~\cite{szczepaniak2020microwave} as shown in Figure~\ref{fig:ttd_internal} because of its advantage over wideband operation. It deploys a series of SPDT switches connected to two transmission lines of variable length, switching to the longer transmission line would increase delay by a certain unit along the path. With $m$ switches, we can generate $m$ bit quantized delay with $2^m$ possible delay values. The phase shift can also be implemented using transmission lines, but transmission line-based phase shift supports only narrowband frequencies. We choose a varactor-based phase shift design for wideband operation.

\noindent
\textbf{$\blacksquare$ Design of FPGA Control Board PCB:} We designed a PCB board that hosts a CMOD A7 FPGA~\cite{cmod35t} to program the delay and phase values in our hardware. The PCB also contains a power distribution module with appropriate LDOs that takes a single 12V input and distributes different voltage levels (e.g., 5V, -5V, and 3.3V) to the phase shifter and TTD chips via a ribbon cable. 

\noindent
\textbf{$\blacksquare$ Beam Pattern Measurement Setup:} 
The whole setup was laid over a motor that was controlled via an automated software program, as shown in Figure~\ref{fig:anechoic_chamber}. To measure the \dpa beam pattern, we rotate the motor $1^\circ$ at a time in the $\pm 60^\circ$ range. We connected the Tx antenna to the VNA's port 1 and the output of the combiner to its port 2. We saved the S21 parameters for each angle over the operating frequency range. Then, we compute the received power using the measured S21 parameters and plot the power versus frequency and angle. We verified the setup in an anechoic chamber and then conducted over-the-air experiments in a lab environment. We used radiation-absorbent materials to isolate the setup to reduce channel effects, which can be misleading and affect the resulting images.

% \input{4_implementation}
% !TEX root = main.tex

\section{Evaluation }
\subsection{Hardware Results}

We obtained over-the-air signal measurements with two beams at various angles and bandwidth parts. Figure~\ref{fig:expt_sim_split_results} shows example beam patterns across frequency and angle with varying bandwidth parts. The two beams point towards $30^o$ and $-30^o$, as can be visualized through the beam patterns. We reduce the bandwidth of one of the beams from 50\% to 20\% (consequently increasing it for the other beam from 50\% to 80\%). For higher split scenarios such as 30-70, 40-60, we drew similar conclusions as 20-80 and 50-50 splits, so we omitted them for brevity. The gain across the frequency band for the smaller bandwidth part is capturing a smaller bandwidth, and vice versa. The gain is high in the desired frequency-angle bin shown by a black box and low elsewhere, as expected. We compared the measured beam patterns with simulated ones to show the similarities. We ensured our simulation plots consider the same quantization levels and beam squint effect we face in our hardware. Figure~\ref{fig:50} shows that the beamforming gain at the desired angle from hardware matches the simulation. However, the hardware patterns suffer from additional degradation due to various factors such as multipath (we did over-the-air experiments in a conference room), interference with other devices such as WiFi at 5 GHz, and hardware artifacts (e.g., the matching effect of wideband transmission lines, PCB substrate, and delay/phase IC). Nonetheless, the general behavior suggests that the desired frequency-angle patterns are achievable with our hardware.

\begin{figure}
    \centering
    \includegraphics[width=\linewidth]{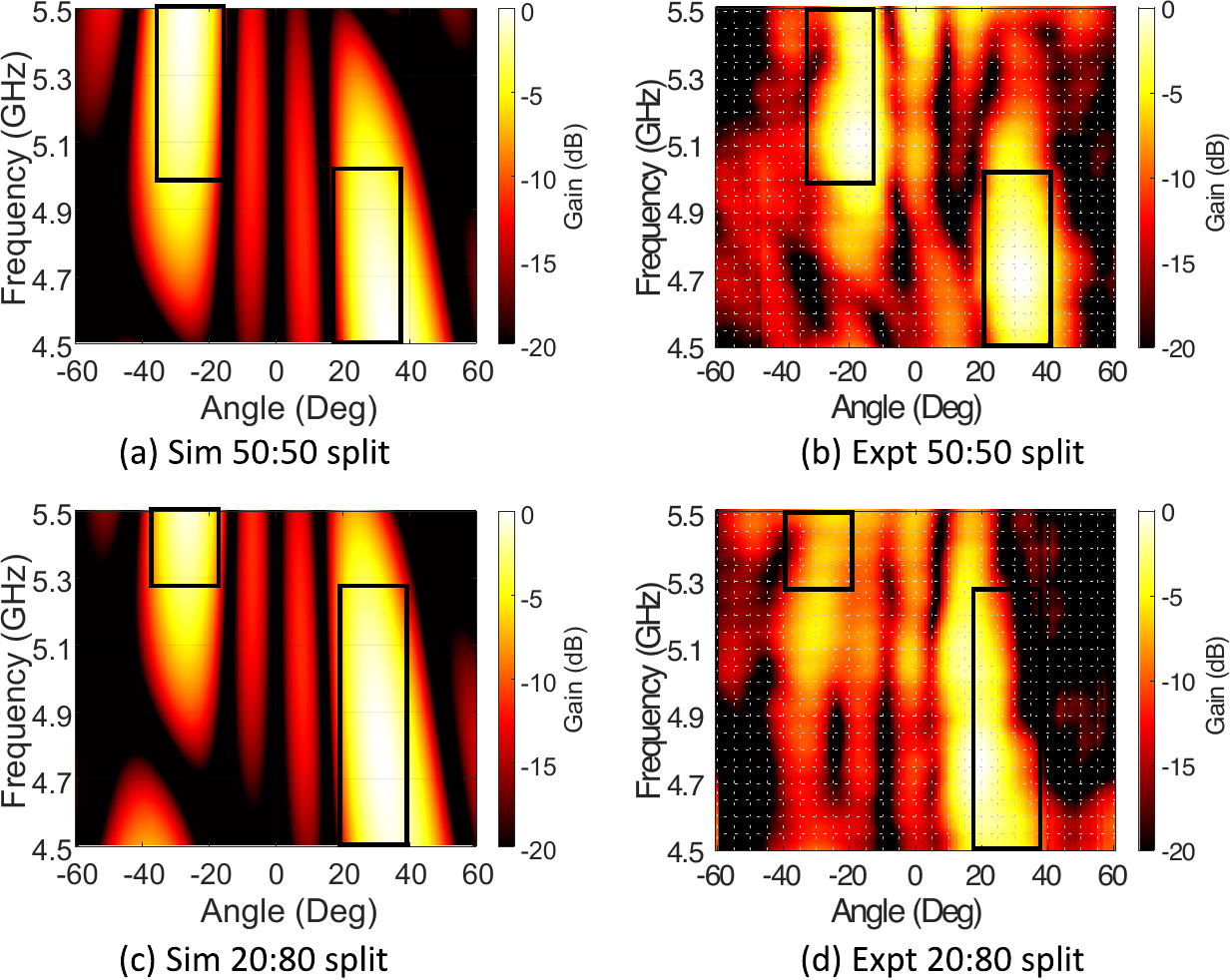}
    \caption{Experimental results of beamforming gain for two beams: $30^o$, $-30^o$ with different bandwidth parts per beam. }
    \label{fig:expt_sim_split_results}
\end{figure}

\begin{figure} [!t]
\centering
    \subfigure[Bandwidth part 50:50]{
        \includegraphics[width=0.22\textwidth]{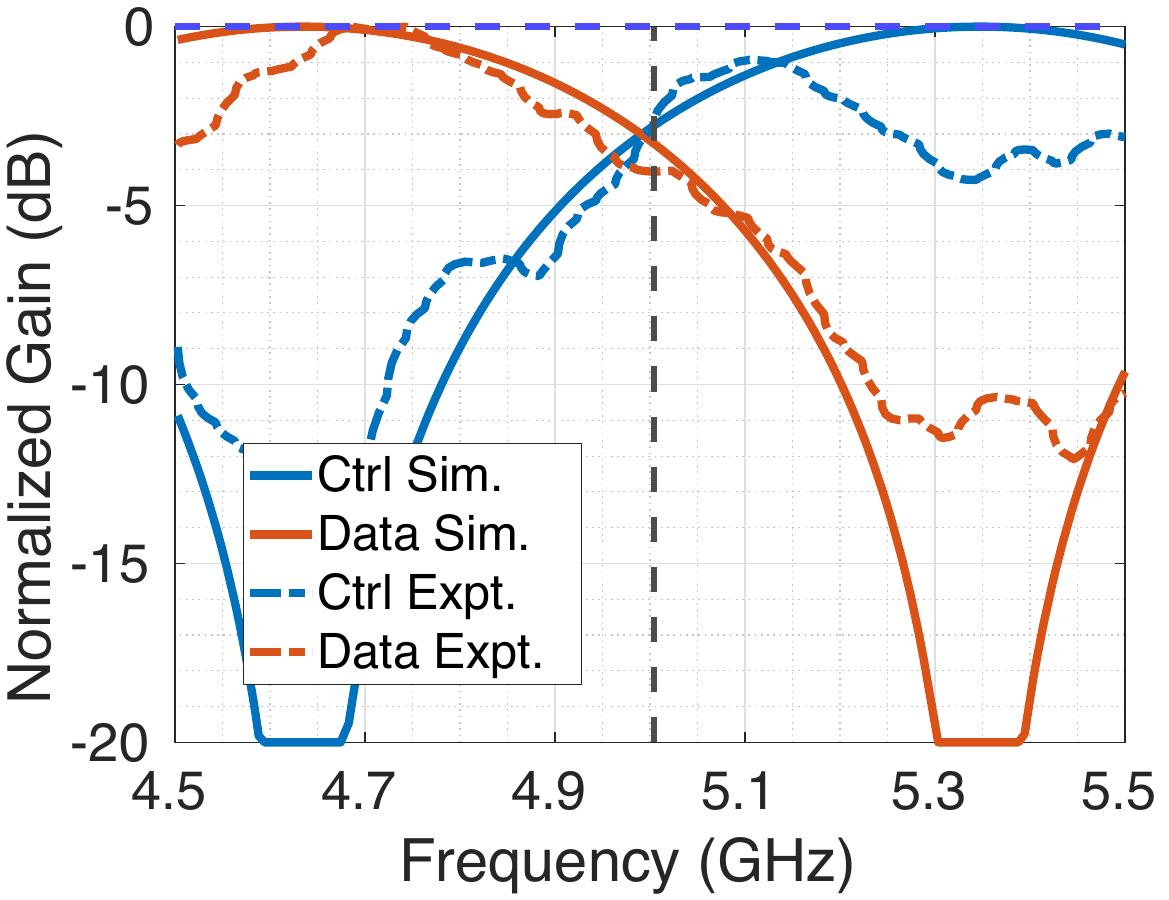}
        \label{fig:50}
      }\hfill
     \subfigure[Gain with subcarrier allocation]{
         \centering
         \includegraphics[width=.22\textwidth]{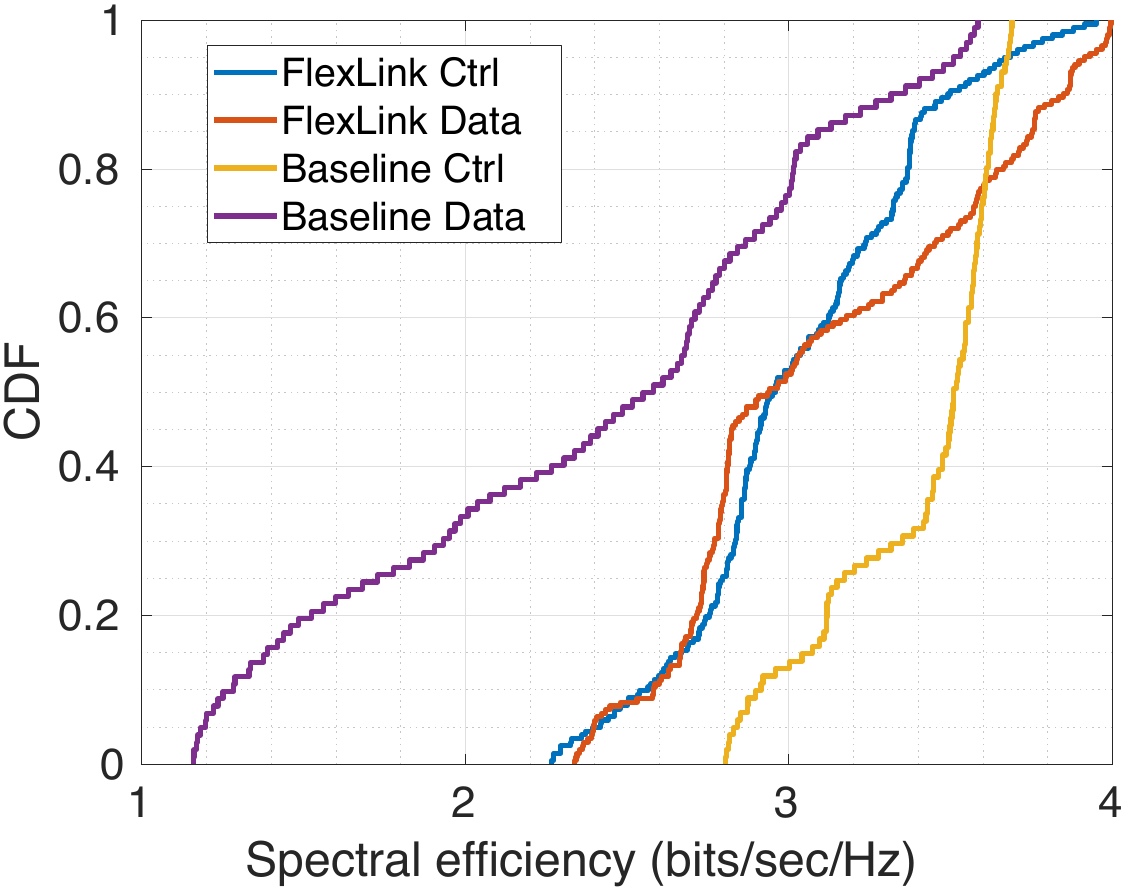}
         % \caption{Gain with subcarrier allocation}
         \label{fig:thput}
     }
     \hfill
    % \subfigure[Gain with number of antennas]{
    %      \centering
    %      \includegraphics[width=.22\textwidth]{figures/flexlink_util_plot.pdf}
    %      % \caption{Gain with number of antennas}
    %      \label{fig:flexlink_util_plot}
    %  }
    \caption{\name improves spectral efficiency and increases spectrum utilization. }
    \vspace{-0.2cm}
 \end{figure}

\noindent
\textbf{Spectral Efficiency Improvement:} 
\name enhances spectral efficiency by enabling high-gain beams for both control and data signals, unlike traditional phased arrays that form a single control beam. As shown in Figure~\ref{fig:thput}, the baseline achieves high efficiency only for control, while \name provides high efficiency for both, nearly doubling data efficiency in the worst 10\% cases. Although control efficiency drops slightly due to linear approximation and hardware artifacts, the overall efficiency remains balanced. Importantly, while the baseline uses only 7\% of spectrum for control (SSB), \name fully utilizes 100\% of the spectrum with both control and data.

\noindent
\subsection{HFSS evaluations}
We use the ANSYS HFSS (High-Frequency Structure Simulator) tool to simulate a \name so that all the non-ideal effects involving substrate waves, material losses, coupling between adjacent elements, and individual antenna element patterns are modeled by providing the appropriate delay and phase values. The element is modeled using a simple inset-fed patch antenna on a Rogers RO4350 substrate with a variable delay + phase applied at the input port (see Figure~(\ref{fig:dpa_hfss_plot}a)). The results of the simulation are shown in Figure~(\ref{fig:dpa_hfss_plot}b) with the main lobe along 0$^{\circ}$ and -30$^{\circ}$.

\subsection{Scaling to large arrays}
Due to the limited 8 antenna array and 1ns max range of delay unit, we cannot perform hardware measurements for a large-scale system. We performed simulations to show the impact of the number of antennas (up to 64 antennas) and bandwidth parts.

\begin{figure}[t!]
    \centering
    \includegraphics[width=0.45\textwidth]{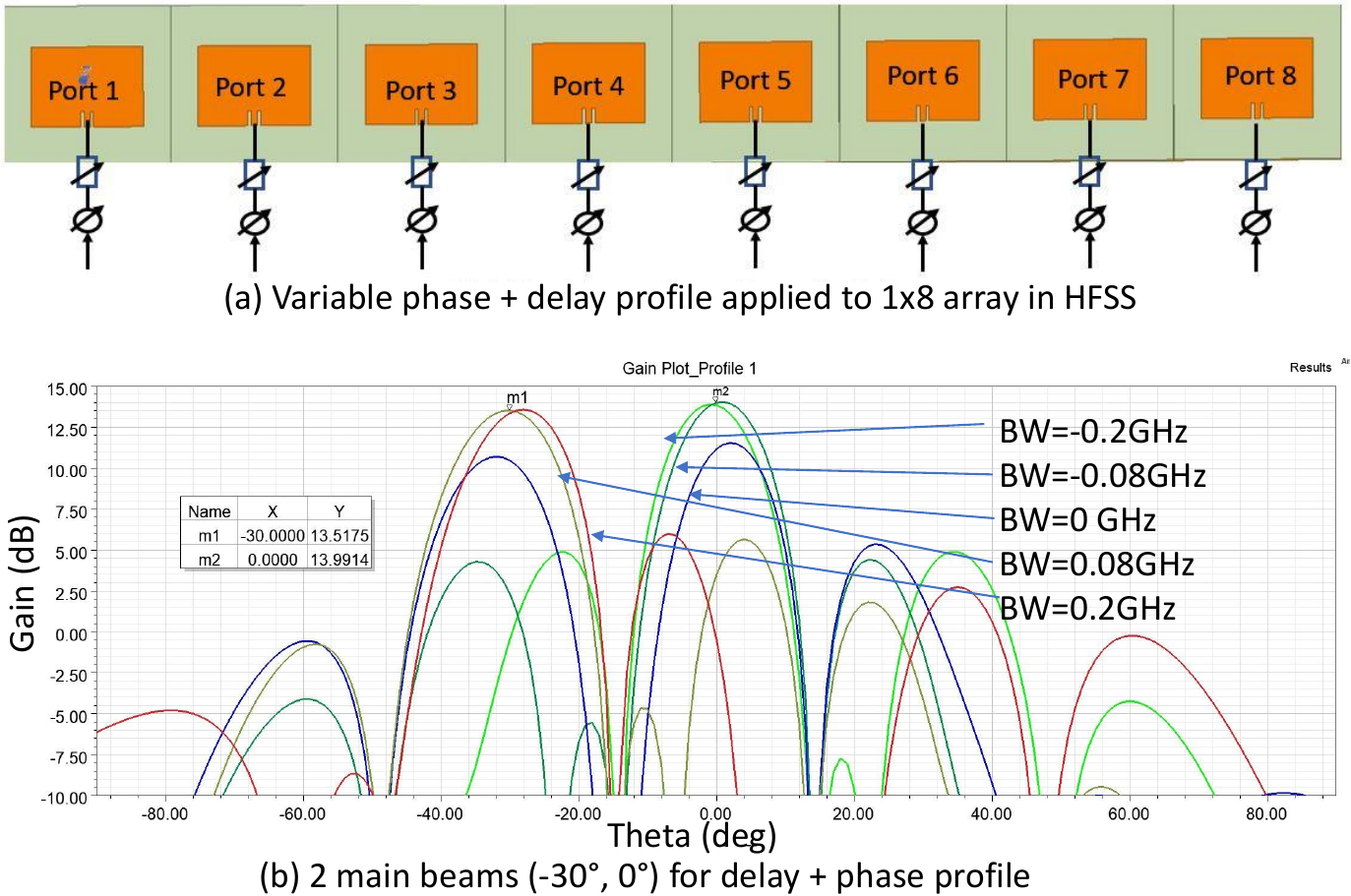}
    \caption{HFSS hardware simulation of \dpa.}
    \vspace{-0.03\textheight}
    \label{fig:dpa_hfss_plot}
\end{figure}

\textbf{Scaling to 3 or more beams:} 
We extensively evaluate \name using a 32-element array and compare it with the high-complexity FSDA algorithm~\cite{jain2023mmflexible}. As shown in Figure~\ref{fig:3beam}(a), \name produces accurate multi-beam patterns for arbitrary beam counts, directions, and bandwidth parts. The antenna gain is high in desired frequency–direction bins, verifying that our linear approximation of the per-antenna phase profile is effective. Only at the frequency edges do we observe gradual transitions, while the beams remain sharp in angular space without leakage to unintended directions.  

% \textbf{Low computation complexity:} 
% \name achieves the same high performance as FSDA while reducing complexity from $O(NM\log(NM))$ to $O(1)$. For instance, with $N=64$ antennas and $M=256$ directions, FSDA requires $\sim$69k multiplications, whereas \name requires fewer than 10 (Figure~\ref{fig:3beam}(b)). This dramatic reduction makes \name highly efficient for FPGA implementation. 

% \textbf{Scaling to 3 or more beams:}
% Here, we provide an extensive evaluation for \name under different scenarios. We compare the beamforming patterns that are produced by our O(1) mathematical expression vs high complexity heuristic-based FSDA algorithm~\cite{jain2023mmflexible}. We utilize a 32-element linear antenna array with programmable phase and delay elements per antenna in various scenarios with three or more beams ( three-beam configuration shown in Figure~\ref{fig:3beam}(a)). 
% \name can reasonably achieve the desired multi-beam patterns with an arbitrary number of beams, beam directions, and bandwidth parts. We observe that the antenna gain is high in the desired frequency-direction pairs and low elsewhere, as expected. This verifies that the linear approximation of the non-linear per-antenna phase profile is fairly accurate. However, at the edges of each beam, there may be a gradual transition instead of a sharp one. This can be seen at the boundary of the frequency band supported by each beam, but it only occurs along the frequency band and not in the angular domain. The beams remain sharp in the intended directions and do not spread to other unwanted angular directions.

\textbf{Low computation complexity:} We show that \name achieves the same high performance with low computational complexity compared to the heuristic algorithm~\cite{jain2023mmflexible}. Figure~\ref{fig:3beam}(b) shows the corresponding delay and phase values obtained by \name and mmFlexible baseline optimization are the same values. These delay and phase values are used to create the beamforming pattern in Figure~\ref{fig:3beam}(a). \name achieves this high performance with a simple O(1) complexity in estimating the delay and phase values, while the baseline requires high $O(NM\log(NM))$ complexity. For instance, $N=64$ antennas and $M=256$ directions, the baseline requires ~69k multiplications, while \name requires less than 10 multiplications, which makes it efficient to program in FPGAs.

\begin{figure*} [!t]
\centering
  \subfigure[3-beam response]{
    \includegraphics[width=.23\textwidth]{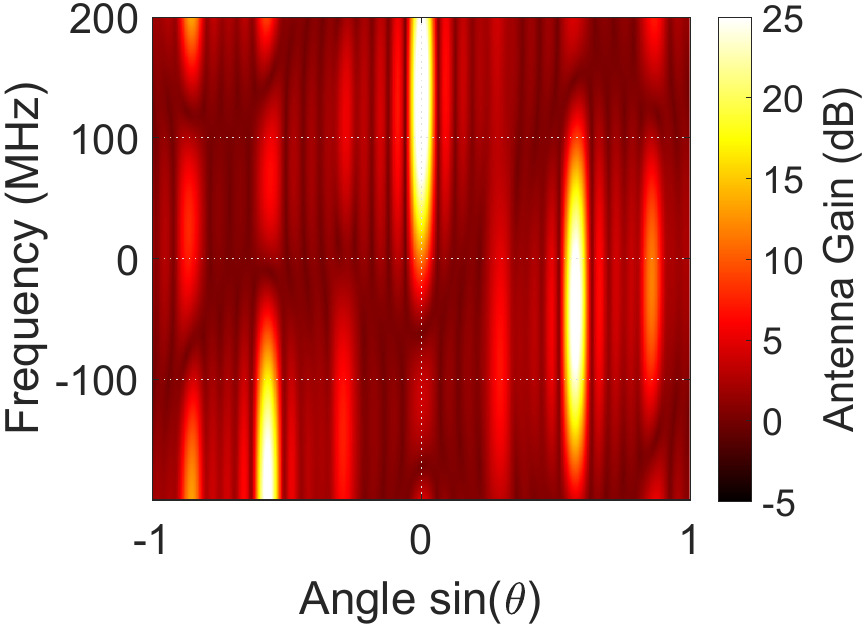}
    \label{fig:image_math_case1}
  }\hfill
  \subfigure[Phases and Delays]{
    \includegraphics[width=.22\textwidth]{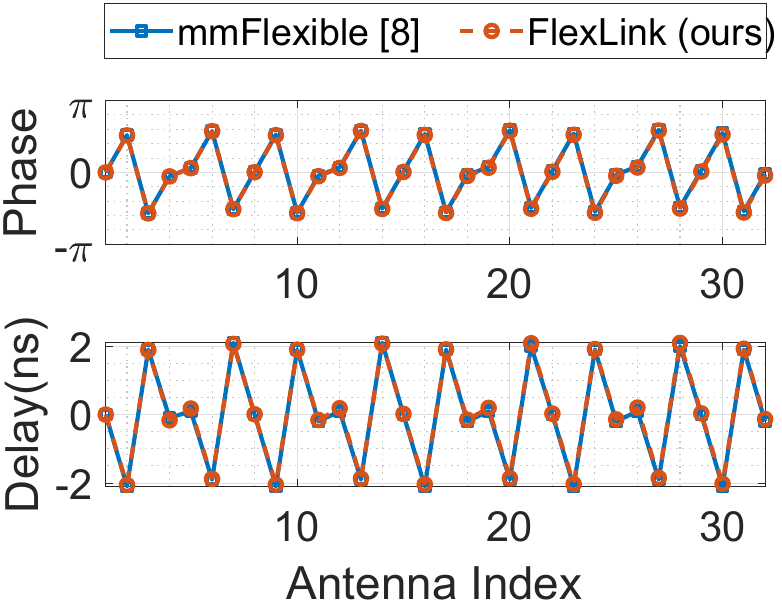}
    \label{fig:phase_delay_beamcase1}
  }\hfill
  \subfigure[Gain with subcarrier allocation]{
         % \centering
         \includegraphics[width=.22\textwidth]{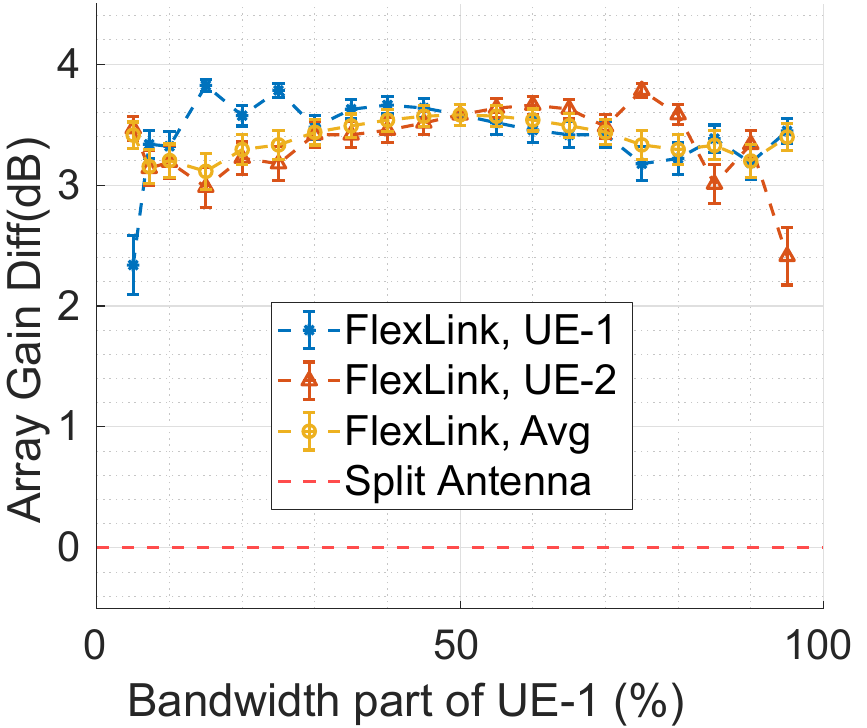}
         % \caption{Gain with subcarrier allocation}
         \label{fig:unequal_allocation}
     }
     \hfill
    \subfigure[Gain with number of antennas]{
         \centering
         \includegraphics[width=.22\textwidth]{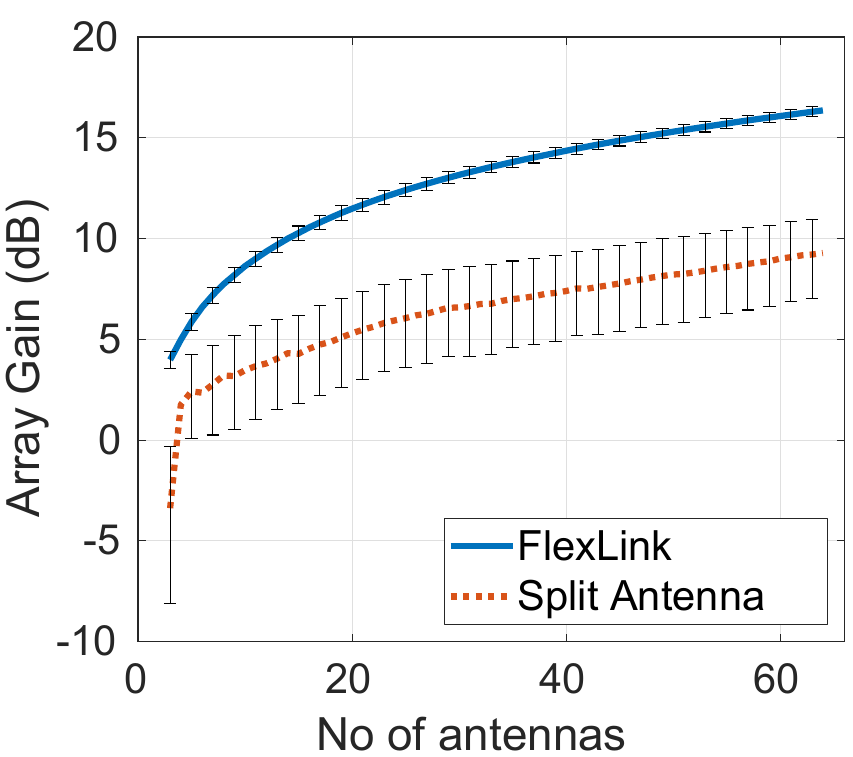}
         % \caption{Gain with number of antennas}
         \label{fig:antenna_variations}
     }
     \vspace{-0.02\textheight}
    \caption{Simulation for large antenna arrays and 3+ number of simultaneous beams. 
    }
    \vspace{-0.015\textheight}
    \label{fig:3beam}
\end{figure*}
\begin{figure}[t!]
    \centering
    \includegraphics[width=0.45\textwidth]{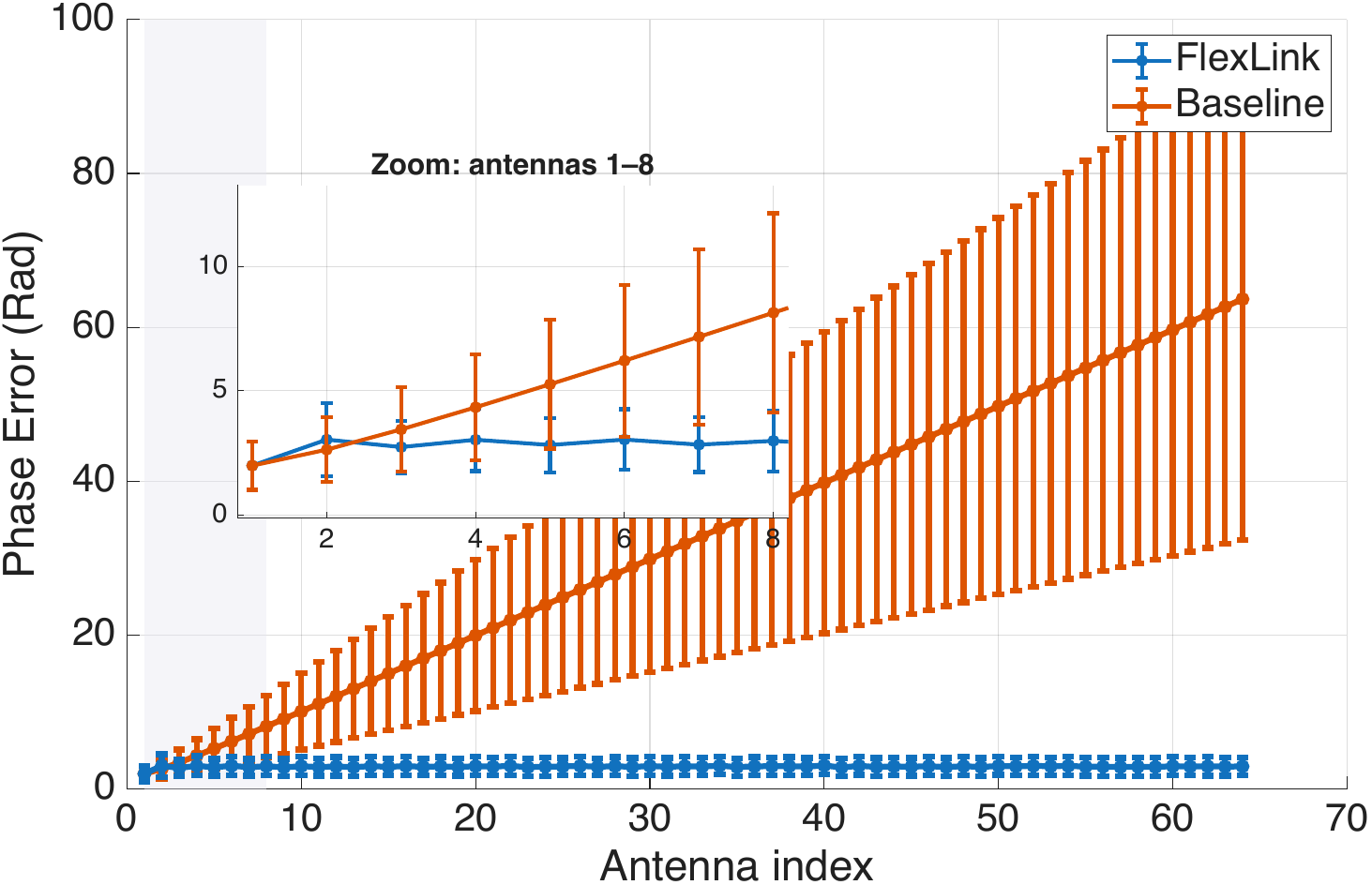}
    \caption{\name reduces phase error in the objective function (\ref{eq:error_objective}) thanks to the innovative phase wrapping technique. }
    \vspace{-0.015\textheight}
    \label{fig:antenna_error_plot}
\end{figure}
\textbf{Impact of bandwidth part allocation:}
Applications such as concurrent communication, control, or scenarios involving both high and low-bandwidth users require consistently high beam gains to ensure low latency and reliability. We show that \name achieves this by maintaining high gain across all users, independent of bandwidth allocation.
To demonstrate, we used a linear 16-element array with $\frac{\lambda}{2}$ spacing and formed two beams in separate directions. If beam 1 (for UE 1) is assigned $x\%$ of sub-carriers, beam 2 (for UE 2) receives the remaining $(100 - x)\%$. We varied $x$ from $5\%$ to $95\%$ and compared gains.
Unlike the split antenna method, which offers constant but lower gain by spanning the full bandwidth in each user direction, \name maintains high and balanced gain, as shown in Figure~\ref{fig:unequal_allocation}. Even at extremes—$5\%$ or $95\%$ allocation—both users experience similar, strong gains. 
Error bars show gain variation due to different beam angle separations. These results confirm that \name supports flexible bandwidth allocation without compromising per-user gain.

% Applications like concurrent communication and control, or use cases with users who require both low and high bandwidth, should have high gains for all beams to user directions to enable low latency and reliable links. Here we demonstrate that \name provides high gains for all users regardless of their bandwidth allocation. To illustrate this, we considered a linear array with 16 antennas and $\frac{\lambda}{2}$ separation and two beams in different directions. Since there are only two beams, if we are giving $x\%$ of sub-carriers to beam 1 (for UE 1), then beam 2 (for UE 2) will get the remaining $(100-x)\%$ sub-carriers. We evaluated gain by varying beam 1 sub-carriers from $5\%$ to $95\%$. The split antenna approach has low constant gain irrespective of sub-carrier allocations, since radiating beams in user directions over the entire bandwidth. Array gain difference with respect to split antenna is shown in Figure~( \ref{fig:unequal_allocation}), demonstrating that even at edge scenarios with beam 1 at $5\%$ sub-carriers or beam 1 at $95\%$, Both users get similar high gains without prioritizing high bandwidth user. 
% The error bars in the respective colors indicate gain variations due to different angle separations between the beams. 
% Thus proving that \name truly enables programmable bandwidth parts without compromising on individual user gain.

\textbf{Scaling with number of antennas:}
Here, we present gain evaluations with an increase in the number of antennas for both \dpa and split antenna techniques. We evaluate both approaches in supporting 3 users in different directions and varying the number of antennas from 3 to 64, as illustrated in Figure~\ref{fig:antenna_variations}, with an increase in the number of antennas, both \dpa and the split technique gain increase with a similar pattern. Note that for the three antennas scenario (to support three user directions), the split antennas approach radiates the same as quasi-Omni. Hence, it has a lower gain than the rest. The error bars in black indicate the variations in gain accommodating different user directions (Monte-Carlo simulations varying user directions). Split antennas as high error bars indicate high gain variations compared to \dpa.

\noindent
\textbf{Significance of phase-wrapping in \name:} Recall that \name uses phase-wrapping to bound the step size in the antenna phase response (Figure~\ref{fig:proof_best_line_fit}). To highlight its impact, we compare against a baseline without wrapping, where step size grows with antenna index. Figure~\ref{fig:antenna_error_plot} shows that \name maintains bounded error ($\sim$3 radians), while the baseline error grows unbounded (8 radians at 8 antennas, 64 radians at 64 antennas). In fact, the resultant beam pattern with this baseline resembles a rainbow beam response~\cite{li2022rainbow}, which spreads all frequencies in all directions uniformly. Thus, phase-wrapping is essential for \name to achieve bounded error independent of array size.

\section{Conclusion and Future work}
% \textbf{$\blacksquare$ Can we use two or more RF chains?}
We presented \name, a system for decoupling control and data signaling through a novel delay-phased array (DPA) antenna architecture. We presented a detailed mathematical analysis of the DPA and derived a closed-form mathematical expression for delay and phase values. We built a hardware prototype for DPA using wideband delay and phase units and showed the feasibility of multi-beam patterns in over-the-air testing and demonstrated its benefits in improving spectrum utilization for control and data signals. 

As a first prototype, \name demonstrates promising capabilities in multi-beamforming and frequency-dependent PHY design. This section discusses remaining system-level challenges and outlines future directions, including mobility support, MAC layer integration, and extension to mmWave frequencies.

\noindent
\textbf{$\blacksquare$ Mobility and dynamic beam switching:} In \name, we focused on a static user with multi-beamforming. Other system-level issues and algorithm designs, such as beam tracking for mobile users, can be adopted from extensive literature in this area. We will discuss how fast \name can perform a beam scan through the control beam while it co-exists with a data beam and compare with traditional approaches, which only use a control beam without any data beam.

\noindent
\textbf{$\blacksquare$ Integration with MAC and higher layers:} \name introduces physical layer optimization tools that incorporate novel hardware and algorithmic approaches for frequency-dependent beamforming. This architectural shift necessitates a reexamination of the end-to-end protocol stack. In particular, a MAC scheduler capable of frequency-directional beam splitting is required to support simultaneous multi-user transmission across distinct spatial directions. While OFDMA-based schedulers have already been adapted to accommodate directional constraints in single-beamforming scenarios, extending these mechanisms to support multi-beamforming remains an important direction for future work.

\noindent
\textbf{$\blacksquare$ Extension to mmWave frequencies:} Implementing a circuit for delay in mmWave frequencies is challenging due to non-linearity, bandwidth, and matching constraints. Therefore, most available delay circuits in mmWave bands are limited to a few picoseconds, which does not meet the requirements of a 1-10 ns delay in DPA. Alternative architectures can be used for mmWave bands where delay elements are implemented at sub-6 frequencies and then upconverted at each antenna using a series of mixers. We leave this exploration for the future.

\section{Acknowledgment}
% We thank YungYi Sun, Sonny Cao, and Nagarjun Bhat for their help with hardware prototyping and experiments, and Prof. Gabriel Rebeiz and Qian Ma from Extreme Waves for providing support with the true-time delay unit. Thanks to the members of the WCSNG lab at UC 
% San Diego for insightful discussion and proofreading the manuscript. The research was supported by NSF awards \# 2211805 and \# 2232481.
We thank YungYi Sun, Sonny Cao, and Nagarjun Bhat for assistance with prototyping and experiments, and Prof. Gabriel Rebeiz and Qian Ma of Extreme Waves for support with the true-time delay unit. We also thank the WCSNG lab at UC San Diego for discussions and proofreading. This work was supported by NSF awards \#2211805 and \#2232481.

\bibliographystyle{unsrt}
\balance
\bibliography{acmart}
\label{lastpage}
% \appendix
% \input{9_appendix.tex}
% \input{sigcomm21_artifact_appendix_template}
\end{document}